 \documentclass[11pt]{article}
\usepackage{amsmath}
\usepackage{graphicx}
\usepackage{algorithm}
\usepackage{algorithmicx}
\usepackage{ifpdf}
\ifpdf\fi
\usepackage[noend]{algpseudocode}
\usepackage{subfigure}
\usepackage{balance}

\usepackage{fullpage}
\usepackage[margin=1in,letterpaper]{geometry}
\usepackage[compact]{titlesec}

\newcommand{\eat}[1]{}

\newtheorem{thm}{Theorem}[section]
\newtheorem{theorem}{Theorem}[section]

\newtheorem{lemma}[thm]{Lemma}

\newtheorem{example}[thm]{Example}

\newtheorem{corollary}[thm]{ Corollary}

\newcommand{\qed}{\hfill \rule{1ex}{1ex}\medskip\\}
\newenvironment{proof}{\paragraph{Proof}}{\qed}

\title{On the Tradeoff between Stability and Fit}

\eat{
\numberofauthors{2}
\author{
\alignauthor Edith Cohen\\
        \affaddr{Microsoft Research}\\
        \affaddr{Mountain View, CA 94043, USA}\\
       \email{edith@cohenwang.com}
\alignauthor Graham Cormode,  Nick Duffield,\\  Carsten Lund\\
        \affaddr{AT\&T Labs--Research}\\
        \affaddr{Florham Park, NJ 07932, USA}
       \email{\{graham,duffield,lund\}@research.att.com}
}
}
\author{
Edith Cohen \thanks{Microsoft Research, SVC (USA)  and Tel Aviv
  University (Israel)}
 \and Graham Cormode \thanks{
  AT$\&$T Labs-Research, 180 Park Avenue, Florham Park, NJ, USA.}
\and Nick Duffield $^\dagger$
\and Carsten Lund $^\dagger$
 \\
{\tt \small edith@cohenwang.com  \{graham,duffield,lund\}@research.att.com}
  } 
\begin{document}

 \date{}

\maketitle

\begin{abstract}
In computing, as in many aspects of life, changes incur cost.  
Many optimization problems are formulated as a one-time instance
starting from scratch. 
However, a common case that arises is when we already have a set of prior
assignments, and must decide how to respond to a new set of
constraints, given that each change from the current assignment comes
at a price. 
That is, 
we would like to maximize the fitness or efficiency of our system, but
we need to balance it with the changeout cost from the previous state.

We provide a precise formulation for this tradeoff and analyze the
 resulting {\em stable extensions} of some fundamental problems in
 measurement and analytics.
Our main technical contribution is a 
stable extension of PPS (probability proportional to size)
weighted random sampling, with applications to monitoring
and anomaly detection problems. 
We also provide a general framework that applies to top-$k$, minimum
spanning tree, and assignment.  
In both cases, we are able to provide exact solutions, and discuss
efficient incremental algorithms that can find new solutions as the
input changes. 
\end{abstract}

\section{Introduction}

Most textbook optimization problems have a clean statement: given a
set of constraints and an optimization target, we are free to choose
any feasible solution, and so can strive for the optimal result. 
However, when applying these techniques, we may find that we do not
begin with a clean slate. 
Rather, we often have some (partial) assignment, perhaps resulting
from a previous optimization over a prior input. 
This assignment is unlikely to be optimal for the current instance,
but each modification to the current assignment can come with a cost. 
We therefore have to balance the improvement from a new assignment
(the fit)
with the cost of changing from the current assignment to the new one
(the stability). 

Consider a network operator who is leasing network connections
to meet customer demand. 
When there are no existing leases, this can be modeled as a standard
graph optimization problem.
But after an initial solution is found, customer demand may change.
The operator then has to solve a more complex optimization,
since there is a (fixed) cost to breaking an existing lease and establishing a
new one. 
This cannot easily be represented as an instance of the original
problem, due to the mixture of pricing schemes.

Formalizing this setting, we target applications where the {\em input
  state} $\boldsymbol{x}\in {\cal X}$ changes over time.
We maintain an {\em output}  
$S\in {\cal S}$ that can be changed in response to changes in
  the input.  
We are interested in the tradeoff between the
{\em fit} (or {\em efficiency}) of the output with respect to the current input and the {\em
  changeout cost} from the previous output.
%
This requirement for `stability' arises in a variety of different
applications:

\begin{trivlist}
\item
{\sl Human-interpretable output:}
In many settings devoted to monitoring and managing a system, there is
a ``dashboard'' that presents information to an operator. 
For example, this could be details of the most congested links in a
network, the heaviest users of the system, 
locations with anomalous activity, and so on. 
The underlying information may be rapidly changing, so that constantly
updating the display would lead
to a confusing mess of information.  
Instead, the operator should be presented with a view which is
reasonably stable, so that they can absorb the information,
and track the change in properties of items over time. 
At the same time, the output should be close to optimal in
detailing the most important items. 

\item
{\sl The overhead of change:}   
In a client-server assignment or a routing problem, changes in the
output may correspond to cache swaps,  job migration,  or
modifications to a routing table. 
In all of these cases, each change incurs a cost, and very rapid
change can lead to poor performance: delays, due to routes changing
midway, or low throughput due to constant reassignment of machines. 
We are therefore willing to settle for a 
solution that is sub-optimal at any instant, but which 
changes only slowly. 

\item
{\sl Audience-aware advertising:}
Consider a roadside electronic billboard that can respond to the 
set of road users passing by. 
We can choose an advertisement to display in response to demographics
of the users, but constantly changing the ad will make it impossible
for any of them to absorb. 

\item
{\sl Allocating Resources to Monitoring:}
Analysis 
such as change or anomaly
detection requires tracking behavior 
of items over time.
In high-throughput systems, it is feasible to collect detailed
statistics on only a representative subset of active items. 
Picking the items to monitor that are currently the most 
`interesting' 
 (top-$k$ or a random sample)
may fail to keep any around long enough to build sufficient
history for particular items, so we prefer to retain monitored items while they are useful. 
Moreover, to efficiently store historic summaries for an extended time period,
we only need to record the
modifications to the summary made over time, so limiting
modifications means that less storage is required.
\end{trivlist}

\medskip
\noindent 
{\bf Contributions:}
We formalize and motivate the algorithmic problem of balancing fitness
and stability.
Our main results concern the {\em stable extension} of some
fundamental problems in measurement and analytics. 
We first consider sampling with stability, and focus on 
PPS (probability proportional to size) weighted random samples.
We then continue to a selection of fundamental optimization problems
which includes matroids
and introduce a unified model for the stable extension of
top-$k$ set,   Minimum Spanning Tree, and Assignment.
An empirical study of our methods demonstrates their advantages
over alternative approaches. 
Last, we point to further classic problems with natural stable extensions.

\begin{itemize}
\item
{\em Stable PPS sample}:
A PPS sample is a weighted random sample where the sampling probability
of each entry is proportional to its weight, but truncated to be at
most $1$ \cite{Hajek1964}.
PPS sampling minimizes the sum of
variances across all distributions with the same expected size sample.
We study PPS as a method to choose a representative selection of
input data (as needed in the motivating applications), and consider
its stable extension. 
In this stable extension, the output is a vector $\boldsymbol{p}$ of inclusion
probabilities which sum to $k$ and the fitness is the respective sum of
variances.  We also maintain a sample
of entries sampled independently according to the current output distribution.
We measure the changeout cost by the expectation
of the set difference of samples between the two sample distributions. 
Under the appropriate procedure, we show that this can be kept to the
minimal value, the $L_1$ difference between the initial and
target distributions. 

\item
{\em Optimization problems with additive fitness function and
  fixed-size output}:   
The simplest problem in this class,  
where there are no further restrictions on the
 output beyond having a fixed size $k$, is to produce the top-$k$ set, containing the $k$ heaviest entries.
In the stable extension, a  valid
output is any subset  of size $k$, where
the fitness of a particular subset
is the  sum of values of entries included in the subset  and the changeout cost
is the set difference between outputs.  The stable extension of
top-$k$ models
monitoring of heavy users or maintaining the most effective cache.
 This framework also applies to
assignment, i.e. minimum bipartite matching (where output must be a
bipartite matching),  and minimum
spanning tree (output must be a spanning tree).

We establish a relation between the stable extension of an
optimization problem in this class
and instances of the original optimization
problem on a parametrically specified modified input.
This relation allows us to apply an algorithm for the original
optimization problem as a black box to obtain a stable solution.
Moreover, we establish
a mapping between a dynamic version of the stable extension (where 
we must efficiently modify the output in response to incremental
updates to the input) and the dynamic version of the original
optimization problem.  This mapping similarly allows 
off-the-shelf
applications of existing dynamic algorithms 
to efficiently recompute a stable solution when the input
is modified.


\end{itemize}

\section{Relation to Time decay models}

A natural first attempt at providing stability of output is
to stabilize the input, which 
is commonly performed using
some time-decayed average (via an exponential, polynomial,
or sliding window decay function) over recent input states~\cite{StatMethods:ebook,CoSt:pods03f}.   
Decay models are successful in capturing the current state when
inputs can be modeled as noisy
samples from a slowly changing underlying  distribution.  In this case,
the time-decayed  (moving) average better captures the ``real'' current state
than the current input.  

When using decay models directly to obtain stability of output, we can compute  
the best-fit output for the time-decayed average.  
This approach lacks, however, when similarity of the input does
not guarantee  similarity, or stability, of the output.

For example, the exact location of optimal cluster centers will vary
even on smoothed input, and so incur the fixed changeout cost. 
Another example is 
a top-$k$ set over an input where $2k$ values are incremented
uniformly and independently at random.  
The top-$k$ set has sensitivity to small perturbation, which
smoothing the data will not eliminate.
Ideally, we would like the output to be stable over
time: the small benefit of choosing the top-ranked values does not
outweigh the cost of output changes. 
Moreover, when the output function
{\em is} insensitive to small perturbations in the input, we cannot fine
tune the tradeoff of fit and stability by controlling the decay
parameter: the stability of the output is not necessarily an increasing
function of the decay parameter.  
Even if we could do so, we would have no analytic guarantees over the tradeoff.
In contrast, our formulation shows that we can find optimal solutions 
for any point on the tradeoff curve. 

Therefore, as we also demonstrate empirically, while time-decay
can be combined with our models as a first phase, (such as to replace
the input by its time-moving average when the latter better captures
the current state), it is not designed to optimally trade off
stability and  fit, and thus, in and of itself,
does not produce
a satisfactory solution to our problem. 

\section{Model}

 The fit of output $S \in {\cal S}$ to input $\boldsymbol{x} \in X$ is  measured by a
{\em fitness function} $\phi: {\cal S} \times X\rightarrow R$.
The {\em best-fit}  $\text{{\sc opt}}(\boldsymbol{x})$ is the output with maximum fitness
for input $\boldsymbol{x}$:
$$
\text{{\sc opt}} (\boldsymbol{x})=\arg\max_{S\in {\cal S}} \phi(S,\boldsymbol{x})\ .$$

Output domains $\cal S$ and corresponding best-fits 
can be assignments (minimum cost assignment),  routing tables
(shortest path routing), subsets of size
$k$ (top-$k$ set), spanning trees (MST), or sampling
distributions (PPS sampling).
The {\em changeout cost}  from output $S_1$ to $S_2$ is 
a
distance function on ${\cal S}\times {\cal S}$: 
$d(S_1,S_2)\geq 0$.

When presented with a new input  $\boldsymbol{x}$,  and a current output $S'$,   
we are interested in a {\em $\Delta$-stable} optimum, which is the best-fit
  output subject to a limit $D$ on changeout cost, and its respective fitness:
\begin{align*}
\Delta\text{-{\sc opt}}(S',\boldsymbol{x},D)   &= \arg\max_{S | d(S,S')\leq D}  \phi(S,\boldsymbol{x}) \\ 
\Delta\text{-}\phi (S',\boldsymbol{x},D) &= \max_{S | d(S,S')\leq D}  \phi(S,\boldsymbol{x}) 
\end{align*}
When fixing $\boldsymbol{x}$ and $S'$, and
varying the allotted change $D$,  we obtain
the optimal tradeoff between fit and changeout cost
using the points $(D,\Delta\text{-}\phi (S',\boldsymbol{x},D))$.

The {\em $\alpha$-stable} optimum $S$ is defined with respect to
 a parameter $a$, which specifies a relation between changeout
cost and improvement in fitness:
\begin{align}
\alpha\text{-{\sc opt}}( S',\boldsymbol{x},a) &= 
 \arg\max_S \big(\phi(S,\boldsymbol{x})-a\ d(S,S')\big) 
 \label{localopt}  \\
\alpha\text{-}\phi (S',\boldsymbol{x},a) &= \phi(\alpha\text{-{\sc opt}}(S',\boldsymbol{x},a), \boldsymbol{x})
 \nonumber
\end{align}
By sweeping the Lagrange multiplier $a$, we can  obtain the tradeoff 
using the points
\begin{equation} \label{tpoint} \big(d(S',\alpha\text{-{\sc opt}}( S',\boldsymbol{x},a)),\alpha\text{-}\phi( S',\boldsymbol{x},a)\big)\ .
\end{equation}

The fitness of the $\alpha$-stable output decreases with
$a$ and that of a 
$\Delta$-optimal one increases with $D$.
For $a=0$, we allow $D$ to be arbitrarily large, and the
$\alpha$-stable optimum is the current best-fit
$\text{{\sc opt}}(\boldsymbol{x})$.  When $a$ is
sufficiently large, $D$ is forced to be small and the $\alpha$-stable
optimum is closer to the previous $S'$.
Two desirable properties satisfied by the problems we
consider are concavity and output monotonicity of the tradeoff.
These properties simplify the computation of the tradeoff and also
mean that the choice of the best tradeoff point is predictable and tunable.

\medskip
\noindent
{\bf Concave tradeoff:}
Concavity means that the function
$\Delta\text{-}\phi (S',\boldsymbol{x},D)$ is  concave in $D$, so that
 the marginal improvement in fitness for 
 increase in allowed changeout cost is non-increasing.  
When the function is continuous
and differentiable in $D$, the marginal improvement is the first derivative
$\frac{\partial  \Delta\text{-}\phi (S',\boldsymbol{x},D)}{\partial D}$ and
concavity means
$\frac{\partial^2  \Delta\text{-}\phi (S',\boldsymbol{x},D)}{\partial^2 D} <
0$.
In this case, when the parameters satisfy 
  $\Delta\text{-{\sc opt}}(S',\boldsymbol{x},D)= \alpha\text{-{\sc
      opt}}(S',\boldsymbol{x},a)$, then
$a = \frac{\partial  \Delta\text{-}\phi (S',\boldsymbol{x},D)}{\partial D}$.

\medskip
\noindent
{\bf Output monotonicity:}
  When the output can be expressed as a vector, monotonicity means
  that each entry (inclusion, sampling probability, etc.) is monotone
when sweeping $a$ and looking at the 
$\alpha$-stable optimum (or equivalently, when sweeping $D$ and looking at the
$\Delta$-stable optimum).
In particular, when output entries are binary,
then for each $i\in S$, there is at most one threshold value $\tau$ so
that $i\not\in \alpha\text{-}\text{{\sc opt}}( S,\boldsymbol{x},a)  \iff a< \tau$ and for each $i\not\in S$
there is at most one threshold value of $\tau$ so
that $i\in  \alpha\text{-}\text{{\sc opt}}( S,\boldsymbol{x},a)  \iff a< \tau$.

\subsection*{Problem formulations}
The applications we target have a sequence of inputs.  A natural 
aim is to seek outputs that are $\alpha$-stable  in each step with
respect to the same parameter $a$ --- meaning
 that we ``price'' the changeout cost the same across steps.
Alternatively or additionally,  we may need to limit the changeout at
each step.

\medskip
\noindent
{\bf Stable solutions and tradeoff:}
The underlying basic algorithmic problem is to compute,
for a given input $\boldsymbol{x}$  and current output
$S$, the respective $\Delta$-stable  and the $\alpha$-stable optimum solutions.

\medskip
\noindent
{\bf Incremental updates:}
 When in each step the input is provided as a small update to the previous one,
 such as  a rewrite, increment, or decrement of a single entry of the
 vector, we would like 
to efficiently compute a new
$\alpha$-stable optimum 
 instead of applying the batch
algorithm to compute it from scratch.

\medskip
\noindent
{\bf Offline problem:}
In the related {\em offline} problem, the
input is a
  sequence $\boldsymbol{x}^{(j)}\in {\cal X}$ ($j\in[N]$) of states and the solution is a sequence
  $S_j\in {\cal S}$ ($j\in [N]$)  of outputs that maximizes
\begin{equation} \label{globalopt}
{\sum_{j\in [N]}} \phi(S_j,\boldsymbol{x}^{(j)})-\alpha {\sum_{j\in[N]}} d(S_j, S_{j-1})\
.
\end{equation}
$S_0$ is defined to be some initial state, or $\emptyset$.
To maximize stability ($\alpha\rightarrow \infty$), the solution is
$S_j\equiv S$ for all $j\in [N]$ where $S=\arg\max_{S\in {\cal S}}
\sum_{j\in[N]} \phi(S,\boldsymbol{x}^{(j)})$\footnote{In
  general, this is {\em not} equivalent to the original optimization
  problem applied to the combination of all inputs.}. 
 For maximum fit,
($\alpha=0$), we use the best-fit at each timestep, $S_j=\text{{\sc opt}}(\boldsymbol{x}^{(j)})$.
 This offline stable optimum is interesting when we seek to minimize global
 change and are willing to tolerate worse fit occasionally in some steps.   We
 would then like an online solution which is closer to the offline stable
optimum, and possibly use a model of the data to achieve this.
  Our formulation of computing a stable solution at each step  is
  appropriate when the aim is to guarantee a continuously good fit, as
  in the motivating scenarios outlined in the introduction.
Note that our model has no prediction component, which
fundamentally sets it apart from other treatments of adapting to a
changing or unknown state including metrical task
systems\cite{BorodinLinialSaks:JACM1992} and regret minimization in
decision theory.

\section{Stable PPS sample}

In this section, 
we consider the stable extension of PPS sampling.  
The input domain ${\cal X}$ consists of real non-negative vectors $\boldsymbol{w}\in R^n_+$
  and the output domain ${\cal S}$  contains probability distributions
  specified by all vectors 
$\boldsymbol{p}\in [0,1]^n$ such that  $\sum_{i=1}^n p_i= k$.
While we modify the distribution,  the actual output is a 
random sample where entry $i$ is sampled with probability $p_i$
and
different entries are sampled independently\footnote{We comment that
  it is often sufficient to only require limited independence, which is easier
to obtain in practice.  
}.

\medskip
\noindent
{\bf Fitness function:}
When a value $w_i$ is sampled with probability $p_i$, the unbiased
non-negative estimator with minimum variance on the weight of the entry is the 
  Horvitz-Thompson estimator, i.e. $\frac{w_i}{p_i}$ when the entry is
  sampled and $0$ otherwise.  
The variance is 
$w_i^2(1/p_i-1)$.

We measure the quality of a sample distribution $\boldsymbol{p}$ for input $\boldsymbol{w}$
 by the sum of variances $\sum_i w_i^2(1/p_i -1)$.  We
can omit the additive term $-\sum_i w_i^2$ (which 
does not depend on $\boldsymbol{p}$) and obtain  the fitness function
$$\phi(\boldsymbol{p},\boldsymbol{w})= -{\sum_i}  \frac{w_i^2}{p_i}\ .$$  

The best-fit for a sample of expected size $k$ is the distribution
with minimum sum of variances. 
This is achieved by PPS (Probability Proportional to Size) sampling
probabilities, $\boldsymbol{p}=\text{pps}(k,\boldsymbol{w})$, 
obtained by solving for $\tau$ the following equation, using $p_i=\min\{1, w_i/\tau\}$:  
\begin{equation} \label{ppstau}
{\sum_i}  \min\{1, w_i/\tau\}  =k\ .
\end{equation}

\begin{algorithm}[t]
\begin{algorithmic}[1]
 \Procedure{Subsample}{$S,\boldsymbol{p},\boldsymbol{q}$}
\For {$i\in [n]$}
\If {$i\not\in S$ and $q_i>p_i$} 
\If {{\sc rand}$()<\frac{q_i-p_i}{1-p_i}$}
\State $S\gets S\cup\{i\}$ \\ 
 \Comment{$i$ is sampled-in with probability $\frac{q_i-p_i}{1-p_i}$}
\EndIf
\EndIf
\If {$i\in S$ and $q_i<p_i$}
\If {{\sc rand}$()> \frac{q_i}{p_i}$}
\State $S\gets S\setminus\{i\}$ \\ \Comment{$i$ is sampled-out
  with probability $1-\frac{q_i}{p_i}$}
\EndIf
\EndIf
\EndFor\\
\Return{$S$}
 \EndProcedure
\end{algorithmic}
\caption{{\sc Subsample} procedure}
\label{alg:subsample}
\end{algorithm}

\medskip
\noindent
{\bf Distance function:}
 The  random sample is a set of entries and our distance
 function measures the {\em expected} set difference between the samples.
The distance therefore depends on the particular {\em subsampling
  procedure} we employ to move from a sample $S$ drawn from
 probability distribution $\boldsymbol{p}$ to a new one drawn from $\boldsymbol{q}$.
Clearly, the distance  must be
at least $\|\boldsymbol{p}-\boldsymbol{q}\|_1 =\sum_i  |p_i-q_i|$. 
We can use a {\sc Subsample} procedure 
(pseudo-code in Algorithm~\ref{alg:subsample}),  which
has distance {\em exactly} $d(\boldsymbol{p},\boldsymbol{q})= \|\boldsymbol{p}-\boldsymbol{q}\|_1 $
and is therefore minimal.  
Subsampling results in input and output samples
that are {\em coordinated}, meaning they are as similar as possible
subject to each conforming to a different distribution.
%
It is possible to modify {\sc Subsample} to use a VarOpt sampling
procedure~\cite{varopt_full:CDKLT10}, which can ensure that the number
of inclusions and exclusions are each between the floor and ceiling of
their expectation. 
Alternatively, we can associate a permanent random number
(PRN)~\cite{BrEaJo:1972} $u_i \in U[0,1]$ so that 
$i \in S \iff u_i \leq p_i$.  
This means that each sample is drawn from the right distribution,
while the samples are coordinated across time steps. 

\begin{example} \label{ex1}
Consider sampling a set of expected size $2$ from $6$ entries.
Suppose we have a current sample $S$, drawn according to a uniform
distribution $\boldsymbol{p}$, where $p_i=1/3$ for $i\in [6]$.
This distribution  has the best fit when the input is uniform.

 Suppose that the new input weights are
$$\boldsymbol{w}=(2,4,1,5,6,0)\ .$$
The PPS sampling probabilities for a sample of size $2$,
obtained with threshold $\tau=9$,  are 
\begin{equation} \label{bestfitex}
\text{pps}(2,\boldsymbol{w})=
\bigg(\frac{2}{9},\frac{4}{9},\frac{1}{9},\frac{5}{9},\frac{2}{3},0
\bigg)\ .
\end{equation}

  The expected changeout cost of modifying the sample from one drawn according to the uniform
$\boldsymbol{p}$ to one drawn according to $\text{pps}(2,\boldsymbol{w})$ is
$\|\boldsymbol{p} - \text{pps}(2,\boldsymbol{w}) \|_1 = \frac{4}{3}$.

The actual changeout to the sampling distribution
$\text{pps}(2,\boldsymbol{w})$
depends on our initial sample $S$ and the randomization when subsampling.
Suppose $S=\{3,4\}$, that is, contains the third and fourth entries.
Entries $\{1,6\}$ are not in $S$ and have reduced sampling probabilities and therefore will not
be included in the new sample. Entry $4$ is in $S$ and has increased inclusion probability and
therefore will be included in the new sample.  Entries $\{2,5\}$ which are not in $S$ but for
which inclusion probabilities increased, need to be sampled for inclusion with
respective probabilities $\{1/6,1/2\}$.  Entry $3$ which is in $S$ but has a reduced inclusion probability is sampled out with probability $2/3$.
Therefore, the new sample will include entry $4$ and will not include entries $\{1,6\}$ but may include any subset of entries $\{2,3,5\}$.

 The stability-fit tradeoff includes 
distributions that are between $\boldsymbol{p}$ and $\text{pps}(2,\boldsymbol{w})$.  These
distributions have higher variance but smaller expected change. In the sequel, we demonstrate
how they are computed on this example.
\end{example}

\subsection{Computing the stability/fit tradeoff}
\label{sec:pps}
In this section we establish the following:

\begin{theorem}
\label{thm:pps}
An $\alpha$-stable PPS distribution, $\Delta$-stable PPS distribution,
and a description of the tradeoff mapping  parameter values to fitness
and vice versa, can be computed in $O(n\log n)$ time.
Moreover, the tradeoff is concave and monotone.
\end{theorem}

Let $\boldsymbol{p}$ be the initial output and $\boldsymbol{w}$ the current
 weights.  
The $\Delta$-stable distribution $\boldsymbol{q}$ is the solution of the
following nonlinear optimization problem:
\begin{align} \label{deltastable}
\text{minimize}
\sum_i \frac{w_i^2}{q_i} \\
\forall i,\  1\geq q_i\geq 0 \nonumber\\
\sum_i q_i = k \nonumber\\
\sum_i |q_i-p_i| \leq D \nonumber
\end{align}

  We are interested in efficiently solving the problem for a
  particular value of the changeout $D$  and also in computing
a compact representation of all solutions.

We reduce \eqref{deltastable} to  two simpler problems.
{\sc Best-increase} computes, for any value $x$, 
 the most beneficial total increase of $x$ to current sampling
 probabilities, i.e., that would bring the biggest 
reduction in the sum of variances.  Similarly, {\sc Best-decrease} computes the least
detrimental total decrease of $x$, i.e., the decrease that results in
the minimum increase of the sum of variances.  
Both problems can be formulated as convex programs.  We
show that descriptions of both functions can be found in time
 $O(n \log n)$.
 Last, we show how to fully describe the
 $\Delta$-stable PPS distributions and $\alpha$-stable distributions
 in time $O(n \log n)$ by using these functions. 

\begin{algorithm}
 \caption{The function $\Delta^+(y)$ and its inverse
   $\underline{\tau}(x)$ 
 for $\boldsymbol{w}$ and $\boldsymbol{p}$\label{deltaplus}}
\begin{algorithmic}[1]
\Procedure{Best-Increase}{$\boldsymbol{w},\boldsymbol{p}$}
\State $W\gets 0$  \Comment{weight of active items}
\State $D \gets 0$  \Comment{Contribution to difference of items
  that exited and negated sum of initial probabilities of active items}
\State $L \gets \emptyset$ \Comment{Initialize $L$ as an empty max-heap}
\For {$i\in [n]$}
\If {$(w_i>0)$}
\If  {$(p_i<1)$}
\State $(A.v,A.t,A.i)  \gets (\frac{w_i}{p_i},\text{"entry"},i)$\\
\Comment{value, breakpoint type, item}
\State $L \gets L \cup A$ \Comment{add entry record to $L$}
\ElsIf {$p_i=0$}  
\State $W\gets W+w_i$ 
\EndIf
\State $(A.v,A.t,A.i)  \gets (w_i,\text{"exit"},i)$\\
\Comment{value, breakpoint type, item}
\State $L \gets L \cup A$ \Comment{add exit record to $L$}
\EndIf
\EndFor
\State $A \gets $ {\sc GetMax}$(L)$ \\ \Comment{Pull Record with largest $A.v$ in
   $L$}
\If {$W>0$} 
\State $\Delta^+ \gets \{([A.v,\infty), \frac{W}{y})\}$ \\\Comment{``Rightmost''
  piece of Function
  $\Delta^+(y)$}
\State $\underline{\tau} \gets \{((0,\frac{W}{A.v}], \frac{W}{x})\}$
\\\Comment{``Leftmost'' piece of Function
  $\underline{\tau}(x)$}
\Else
\State $\Delta^+ \gets \emptyset$,  $\underline{\tau} \gets \emptyset$
\EndIf
\State $r \gets A.v$ \Comment{definition boundary so far of $\Delta^{+}$}
\State $\ell \gets \frac{W}{A.v}$  \Comment{definition boundary so far of $\underline{\tau}$}

\While {$L\not=\emptyset$}
\State $A \gets $ {\sc GetMax}$(L)$
\State $i\gets A.i$
\If {$A.t = \text{"entry"}$}
\State $W\gets W+w_i$
\State $D\gets D-p_i$
\Else \Comment{$A$ is an exit point}
\State $W\gets W-w_i$
\State $D\gets D+1$
\EndIf
\State $\Delta^+ \gets \Delta^+ \cup ([A.v,r],D+\frac{W}{y})$ \\\Comment{Function $\Delta^+$ from $A.v$ to next
  breakpoint}
\State $r \gets A.v$
\State $\underline{\tau} \gets \underline{\tau} \cup \{((\ell, D+\frac{W}{A.v}],  \frac{W}{x-D})\}$ \\\Comment{Function
  $\underline{\tau}$ from previous to current breakpoint}
\State $\ell \gets D+\frac{W}{A.v}$
\EndWhile
\\\Return {$\Delta^+$, $\underline{\tau}$ }
\EndProcedure
\end{algorithmic}
\end{algorithm}

\subsubsection*{The most beneficial increase}

 The most beneficial increase for $x\in (0,\sum_i (1-p_i)]$ is $\boldsymbol{\delta}$
 which solves the following convex program.
\begin{align} \label{xincrease}
\text{minimize}
\sum_i \frac{w_i^2}{p_i+\delta_i} \\
\forall i,\  1-p_i\geq \delta_i\geq 0 \nonumber\\
\sum_i \delta_i = x \nonumber
\end{align}

\begin{lemma} \label{inclemma}
For any increase $x\in (0,\sum_i (1-p_i)]$ there is a threshold value
$y\equiv \underline{\tau}(x)$
such that the solution of \eqref{xincrease} is
$\delta_i = \min\{1, \max\{p_i,w_i/y\}\} -p_i$.
The function $\underline{\tau}(x)$ and its inverse
$\Delta^{+}(y)$ are continuous and decreasing and piecewise with at
most $2n$ breakpoints, where each piece is a hyperbola.
Moreover, a representation of both functions can be computed in
$O(n\log n)$ time.
\end{lemma}
\begin{proof}
From convexity, the solution is unique and the
 Kuhn-Tucker conditions specify its form:
The partial derivatives
$$\frac{\partial \frac{w_i^2}{p_i+\delta_i}}{\partial
  \delta_i}=-\frac{w_i^2}{(p_i+\delta_i)^2}\ $$  are equal for all
$i$
such that $1-p_i > \delta_i > 0$. 
We denote this common value of $\frac{w_i}{p_i+\delta_i}$ by
$\underline{\tau}(x)$.
For $i$ where $\delta_i=0$, we have
$\frac{w_i}{p_i} \leq \underline{\tau}(x)$ and when $0<\delta_i=1-p_i$,
we have $w_i \geq \underline{\tau}(x)$.

 Using this, we can now consider the solution as a function of $x$.
We first sort all items 
with $0<p_i<1$ in decreasing order of $w_i/p_i$ and place items with 
$p_i=0$ and $w_i>0$ at  the head of the sorted list.  
The list is processed in order: we first
 increase $\delta_1$ until $\frac{w_1}{p_1+\delta_1}=\frac{w_2}{p_2}$.  We then
 increase both $\delta_1$ and $\delta_2$ so that the ratios are all equal until
 $\frac{w_1}{p_1+\delta_1}=\frac{w_2}{p_2+\delta_2}=\frac{w_3}{p_3}$ and so on (but we stop
 increasing a probability once it is equal to $1$).  At any point, the
 ratios $\frac{w_i}{p_i+\delta_i}$  on the prefix $1,\ldots,h$  we have processed are all
 equal to some value $y$  except for
 items where $p_i$ reached $1$.
The threshold $y\equiv\underline{\tau}(x)$ can be interpreted as the
rate of decrease of the 
sum of variances per unit increase in expected sample size after implementing optimally a total
increase of $x$.   It is a decreasing function of $x$ and
satisfies $y\in (w_{h+1}/p_{h+1},w_h/p_h]$, where $h$ is the last
processed entry on our sorted list.
 The new sampling probabilities $\boldsymbol{q}=\boldsymbol{p}+\boldsymbol{\delta}$ 
satisfy $q_i(y)=\min\{1, \max\{p_i,w_i/y\}\}$.
Observe that the probabilities $\boldsymbol{q}(y)$ are  PPS of size $\sum_{i=1}^h
p_i+x$ of the prefix of the list
$$(q_1,\ldots,q_h)\gets \text{pps}({\sum_{i=1}^h}
p_i+x,(w_1,\ldots,w_h))\ ,$$ and $y$ is the
 threshold value of these PPS
probabilities.  

Consider
the (at most $2n$) points defined by the ratios $w_i/p_i$ (when $p_i<1$
and $w_i>0$) and the
weights $w_i$.   Consider the relation between $x$ and $y=\underline{\tau}(x)$.
The function is a hyperbola $x=a+b/y$ for some
constants $a,b$ between consecutive breakpoints.

The inverse function of $\underline{\tau}(x)$, which we call
$\Delta^+(y)$, maps threshold values to probability increases, and
satisfies
\begin{align*}
 x=\Delta^+(y)& ={\sum_i} q_i(y)-p_i\\ 
  &={\sum_i} \min\{1, \max\{p_i,w_i/y\}\}-p_i\ .
\end{align*}
It is continuous and decreasing
with
range $[\min_{i | p_i<1} w_i,R)$, where 
$R=\infty$ if $\exists i$ with $w_i>0$ and $p_i=0$ and
$R=   \max_{i| p_i<1} \frac{w_i}{p_i}$ otherwise.  
It is piecewise 
where each piece is a hyperbola.

Algorithm~\ref{deltaplus} computes a representation of the functions
$\Delta^+(y)$ and  $\underline{\tau}(x)$ each as
a sorted list of pieces of the form $(I,f(y))$ where $I$ are
intervals that partition the domain and $f$ is the function for
values on that interval.
 For $\Delta^+(y)$, the pieces have the form $f(y)=a+b/y$.
The inverse function $\underline{\tau}(x)$ is simultaneously
produced by the algorithm
by inverting each piece $(I,a+b/y)$ to obtain a respective interval where end
point $v$ of $I$ is mapped to end point $f(v)=a+b/v$ (left end points are
mapped to right end points and vice versa).  
The inverse function on the corresponding interval is
 $\underline{\tau}(x)=f^{-1}(x)=\frac{b}{x-a}$.

For each $i$ with $w_i>0$ and $p_i<1$, Algorithm~\ref{deltaplus} creates an {\em entry
    point} with value $\frac{w_i}{p_i}$ and for each item with
  positive $w_i$, creates an exit point with value $w_i$.
These (at most $2n$) points are the breakpoints of $\Delta^+(y)$.
The breakpoints are processed  in order to
compute the functions $\Delta^+(y)$ and $\underline{\tau}$.
The computation is performed in linear time after sorting the
breakpoints.
\end{proof}

\begin{example}
Returning to our example (Example~\ref{ex1}), we consider the best
increase for $\boldsymbol{w}$ and the uniform $\boldsymbol{p}$.  Since probabilities
are uniform, 
 Algorithm~\ref{deltaplus} processes 
entries in order of decreasing weight.
The first entry on our list is entry $5$, which has weight of $6$.
The best increase is applied 
 only to this entry until the ratio matches the ratio of the
second entry on the list, which happens when 
 $6/(1/3+\delta)=5/(1/3)$, that is,  when $\delta=1/15$.
So the first piece of our function is for the interval $x\in (0,1/15]$.  
The corresponding distribution has entry $5$ with probability $1/3+x$
(and all other entries with the original $1/3$).   The algorithm has
active weight $W=6$ and thus produces the 
threshold function $\underline{\tau}=6/(1/3+x)$.
We can verify that the inclusion probability for changeout $x$ is $q_5=6/\underline{\tau}(x)=1/3+x$.
 The second piece of the function associates increases for both entry $5$, and the second entry on our list, which is entry $4$.  They are increased until the ratios are equal to the third entry on our list, which is entry $2$.  This happens when the increase to entry $5$ is $1/6$, according to $6/(1/3+\delta)=4/(1/3)$ and
when the increase to entry $4$ is $1/12$ according to $5/(1/3+\delta)=4/(1/3)$. The total
increase from entries $\{4,5\}$ is $1/6+1/12=1/4$ and we obtain that the second piece of the function is for $x\in (1/15,1/4)$, and both entries $\{4,5\}$ are increased.
 Algorithm~\ref{deltaplus} computes 
the active weight  $W=6+5$ and the
threshold $\underline{\tau}(x)=11/(x+2/3)$.  Thus for an increase of $x$, entry $5$ is sampled with probability $6/\underline{\tau}(x)=(6/11)(x+2/3)$ and entry $4$ with probability
$5/\underline{\tau}(x)=(6/11)(x+2/3)$.  The third piece would have $3$ active entries $\{2,4,5\}$ (the three largest entries.).  The active weight is $W=15$ and the threshold is accordingly 
$\underline{\tau}(x)=15/(x+1)$.  The interval for this third piece is $(1/4,2/3]$.
\end{example}


\subsubsection*{The least detrimental decrease}
We next address the complementary problem, of how to compute the best
decrease to the probabilities. 
For $x\in [0,\sum_i p_i]$, it is the solution of the optimization problem
\begin{align} \label{xdecrease}
\text{minimize}
\sum_{i| w_i>0} \frac{w_i^2}{p_i-\delta_i} \\
\forall i,\  p_i\geq \delta_i\geq 0 \nonumber\\
\sum_i \delta_i =  x \nonumber
\end{align}

We first observe that 
any total decrease $x\leq D_0$,  where
 $D_0=\sum_{i|w_i=0} p_i$ can be arbitrarily applied to
probabilities of items with zero weight and does not increase the
sum of variances.  

\begin{lemma} \label{declemma}
For any decrease $x\in (D_0,\sum_i  p_i]$,    there is a  threshold  
$y\equiv \overline{\tau}(x)$ such that    the solution of \eqref{xdecrease}
has the form  $\delta_i=p_i-\min\{p_i,w_i/y\}$. 
The function $\overline{\tau}(x)$   and its inverse $\Delta^-(y)$ are
increasing, continuous, and piecewise with at most $n$ breakpoints where each
piece is a hyperbola.
Moreover, a representation of both functions can be computed in
$O(n\log n)$ time.
\end{lemma}
\begin{proof}
For  any decrease $x\in (D_0,\sum_i  p_i]$,   the solution is unique and
specified by the Kuhn-Tucker conditions.
For all $i$ where $w_i=0$ we
must have $\delta_i=p_i$.
For all $i$ such that
$\delta_i \in (0,p_i)$, the partial derivatives 
$$\frac{\partial \frac{w_i^2}{p_i-\delta_i}}{\partial
  \delta_i}=-\frac{w_i^2}{(p_i- \delta_i)^2} $$ are equal and 
we denote the common value of
$\frac{w_i}{p_i- \delta_i}$ by $\overline{\tau}(x)$.
When $\delta_i=0$, we must have  $\frac{w_i}{p_i} \geq \overline{\tau}(x)$.

The solution thus has the form
$$\delta_i = p_i- \min\{p_i, w_i/y\}= \max \{0, p_i-w_i/y\}\ .$$
We obtain the relation
$$x=\sum_i \delta_i = \sum_{i | w_i/y < p_i} (p_i - w_i/y)\ ,$$ which defines
the functions
$y\equiv \overline{\tau}(x)$  and its inverse
$x\equiv \Delta^-(y)$.
Both $\overline{\tau}(x)$ and its inverse
$\Delta^-(y)$ 
are continuous and decreasing.  The range
of $\overline{\tau}(x)$ (domain of $\Delta^-(y)$) is
$[\min_{i | p_i>0} w_i/p_i,\infty)$. 
Both are piecewise  with
 at most $n$ breakpoints which correspond to the values of the 
 ratios $w_i/p_i$.
Algorithm~\ref{deltaminus} computes a representation of
these two  functions.
 It first  orders the items as in Algorithm~\ref{deltaplus} and
 the proof of Lemma~\ref{inclemma}, according to the ratios $w_i/p_i$
 and processes this list in reverse order, generating the pieces of the two functions.
Each function is represented
as  a list of
consecutive pieces  of the form
$(I,f(y))$, where $I$ is an interval.  For convenience, we define $\overline{\tau}(x)=0$  for $x
\leq  D_0$. 
For  $\Delta^-(y)$, $f(y)$ has the form $f(y)=a-b/y$.  
The algorithm inverts $f(y)$ to obtain the corresponding piece of the function $\overline{\tau}(x)$:
an end point $v$ of $I$ is mapped to end point
$a-b/v$ (left end points are mapped to right end points and vice
versa) and the function is $b/(a-x)$.

\eat{
Items are ordered as in the proof of Lemma~\ref{inclemma},  defining
at most $n$ breakpoints.  We then 
scan the list in reverse order.  We first increase $\delta_n$
 while $\frac{w_n}{p_n-\delta_n}\leq \frac{w_{n-1}}{p_{n-1}}$,
then similarly increase both $\delta_n$ and $\delta_{n-1}$ while
$\frac{w_n}{p_n-\delta_n} = \frac{w_{n-1}}{p_{n-1}-\delta_{n-1}}\leq 
\frac{w_{n-1}}{p_{n-2}-\delta_{n-2}}$ and so on.

For a given $x$, we can consider the suffix $(t,\ldots,n)$ of
  items with probabilities modified by this process.
 Items with $w_i=0$, if there are any, are at the tail of the list,
 and clearly, for these items, 
we  can increase $\delta_i$ (decrease inclusion probabilities) without increasing variance.  
Thus when the decrease $x \leq D_{o} \equiv \sum_{i|w_i=0} p_i$,
we only modify probabilities of items with zero weights.
When $x > D_0$, then all items with $w_i=0$ have zero
probabilities and
all items in the suffix with $w_i>0$ have the same ratio
$w_i/p_i$ which is  equal to some $y  \in
[w_t/p_t, w_{t-1}/p_{t-1})$.
We therefore have
$\delta_i= p_i-\min\{p_i, w_i/y\}$.
}
The probabilities
$\boldsymbol{q}\gets\boldsymbol{p}-\boldsymbol{\delta}$
computed for a given $x$ are PPS with threshold is $y$:
$$(q_t,\ldots,q_n) \gets \text{pps}(\sum_{i=t}^n p_i -
x,(w_t,\ldots,w_n))\  .$$
\end{proof}

\begin{algorithm}[t] 
 \caption{The function $\Delta^-(y):(0,\infty)$ and its
   (extended) inverse
   $\overline{\tau}:(0,\sum_i p_i)$ for $\boldsymbol{w}$ and $\boldsymbol{p}$\label{deltaminus}}
\begin{algorithmic}[1]
\Procedure{Best-Decrease}{$\boldsymbol{w},\boldsymbol{p}$}
 \State $W\gets 0$  \Comment{Initial active weight $W$}
 \State $D \gets 0$  
 \State $\Delta^{-}\gets \emptyset$, $\overline{\tau} \gets \emptyset$
\For {$i\in [n]$}
 \If {$(w_i>0)$} 
 \State $(A.v,A.i)  \gets (\frac{w_i}{p_i},i)$
 \Comment{value, item}
 \State $L \gets L \cup A$ \Comment{add entry record to $L$}
 \Else $,\ D\gets D+p_i$    \Comment{$w_i=0$}
 \EndIf
 \EndFor
 \State $\overline{\tau} \gets ((0,D],0) $ \Comment{Threshold function is
   $0$ for $y\leq D$}
\State $\ell=D$  \Comment{Left boundary of current definition of
  $\overline{\tau}$}
\State $r=\infty$ \Comment{Right boundary of current definition of $\Delta^{-1}$}
\While {$L\not=\emptyset$}
 \State $A \gets $ {\sc GetMin}$(L)$ \\\Comment{Pull out record with minimum $A.v$}
 \State $i\gets A.i$
 \State $W\gets W+w_i$
 \State $D\gets D+p_i$
 \State $\Delta^{-} \gets \Delta^{-} \cup  ([A.v,r), D-\frac{W}{y})$
\State $r \gets D-\frac{W}{A.v}$
 \State $\overline{\tau} \gets \overline{\tau} \cup ((\ell, D-\frac{W}{A.v}],
 \frac{W}{D-x}) $
\State $\ell \gets D-\frac{W}{A.v}$
 \EndWhile
\EndProcedure
\\\Return {$\Delta^-$, $\overline{\tau}$ }
\end{algorithmic}
\end{algorithm}

\begin{example}
We compute the best decrease for Example~\ref{ex1}, as computed by
 Algorithm~\ref{deltaminus}. We obtain that for $x\in (0,1/3]$ we only decrease the
probability for entry $6$, which has weight of $0$.  We then decrease only the second
smallest entry, which is entry $3$ with weight $1$.  The interval of the second piece 
is accordingly
$x\in (1/3,1/2]$, we have active weight $1$, $\overline{\tau}(x)=
1/(2/3-x)$, and $q_6=0$ and $q_2=1/\overline{\tau}(x)=2/3-x$.  The third piece $x\in (1/2,2/3]$
involves entries $\{1,3\}$.  The active weight is $W=3$ and $\overline{\tau}(x)=\frac{3}{1-x}$.
Accordingly, the probabilities are
$q_6=0$, $q_2= 2/\overline{\tau}(x)$, and $q_1=1/\overline{\tau}(x)$.
\end{example}

 \subsubsection*{Computing a $\Delta$-stable PPS distribution}
  We consider computing
$\Delta\text{-{\sc Opt}}(\boldsymbol{w},\boldsymbol{p},D)$, which is the distribution with 
maximum fitness (minimum variance) $\phi(\boldsymbol{q},\boldsymbol{w})$, amongst $\boldsymbol{q}$ that are within distance $D$ from
$\boldsymbol{p}$ ($\|\boldsymbol{p}-\boldsymbol{q}\|_1 \leq D$).
The applicable values of $D$ are in
the interval  $[0, \|\text{pps}(k,\boldsymbol{w})-\boldsymbol{p}\|_1]$
and clearly
$\Delta\text{-{\sc Opt}}(\boldsymbol{w},\boldsymbol{p},0)=\boldsymbol{p}$ and
$\Delta\text{-{\sc Opt}}(\boldsymbol{w},\boldsymbol{p},\|\text{pps}(k,\boldsymbol{w})-{\boldsymbol{p}}\|_1)=\text{pps}(k,\boldsymbol{w})$.
For general $D$, we obtain the distribution by combining the best
increase of $D/2$ and the best decrease of
$D/2$. The two are disjoint for $D\leq \|\text{pps}(k,\boldsymbol{w})-\boldsymbol{p}\|_1$.
Pseudocode is provided in Algorithm~\ref{optdelta}.

\begin{example}
On Example~\ref{ex1}, the optimal PPS has changeout $4/3$ and thus we are interested in
the  $\Delta$-stable distribution for $D< 4/3$.  We do this by combining the best increase
of $D/2$, which affects a prefix of the entries $\{5,4,2\}$ with the best decrease of $D/2$, 
which affects a prefix of the entries $\{6,3,1\}$.
If $D=1$, the best increase of $x=0.5$ involves all three entries.  We have $\underline{\tau}=7.5$
and accordingly the entries $\{5,4,2\}$, which have weights $\{6,5,4\}$, have
$q_5=6/7.5= 0.8$, $q_4=5/7.5=2/3$, and $q_2=4/7.5=8/15$.  The best decrease
of $x=0.5$ has $\overline{\tau}=6$.  Accordingly, $q_6=0$, $q_3=1/6$, and $q_1=2/6=1/3$.
Therefore, the $\Delta$-stable distribution for $D=1$ is
$$\boldsymbol{q}=\bigg(\frac{1}{3},\frac{8}{15},\frac{1}{6},\frac{2}{3},\frac{4}{5},0\bigg)\
.$$
\eat{

gnuplot
set terminal postscript eps color 24
set output "thresholds.eps"
set xlabel 'changeout'
set ylabel 'threshold value'
plot [0:0.666] (x< 0.0667) ? 6/(0.333+x) : x<0.25 ? 11/(0.667+x)
:15/(1+x) title 'x increase threshold' lw 3, (x< 0.0333) ? 0 : x<0.5 ? 
1/(0.667-x) : 3/(1-x) lw 3 title 'x decrease threshold', 9 lw 2 title 'pps threshold'

}
Figure~\ref{threshex:assign}  illustrates the thresholds
$\overline{\tau}$ and $\underline{\tau}$ as a function of the
changeout $x$.  When the changeout is $D=4/3$, the increase and
decrease are $D/2=2/3$ and 
both thresholds
meet the PPS threshold, reflecting the fact that the
$\Delta$-stable solution is the best-fit PPS \eqref{bestfitex}.
\end{example}

\begin{figure}[t]
\centering
 \includegraphics[width=0.3\textwidth]{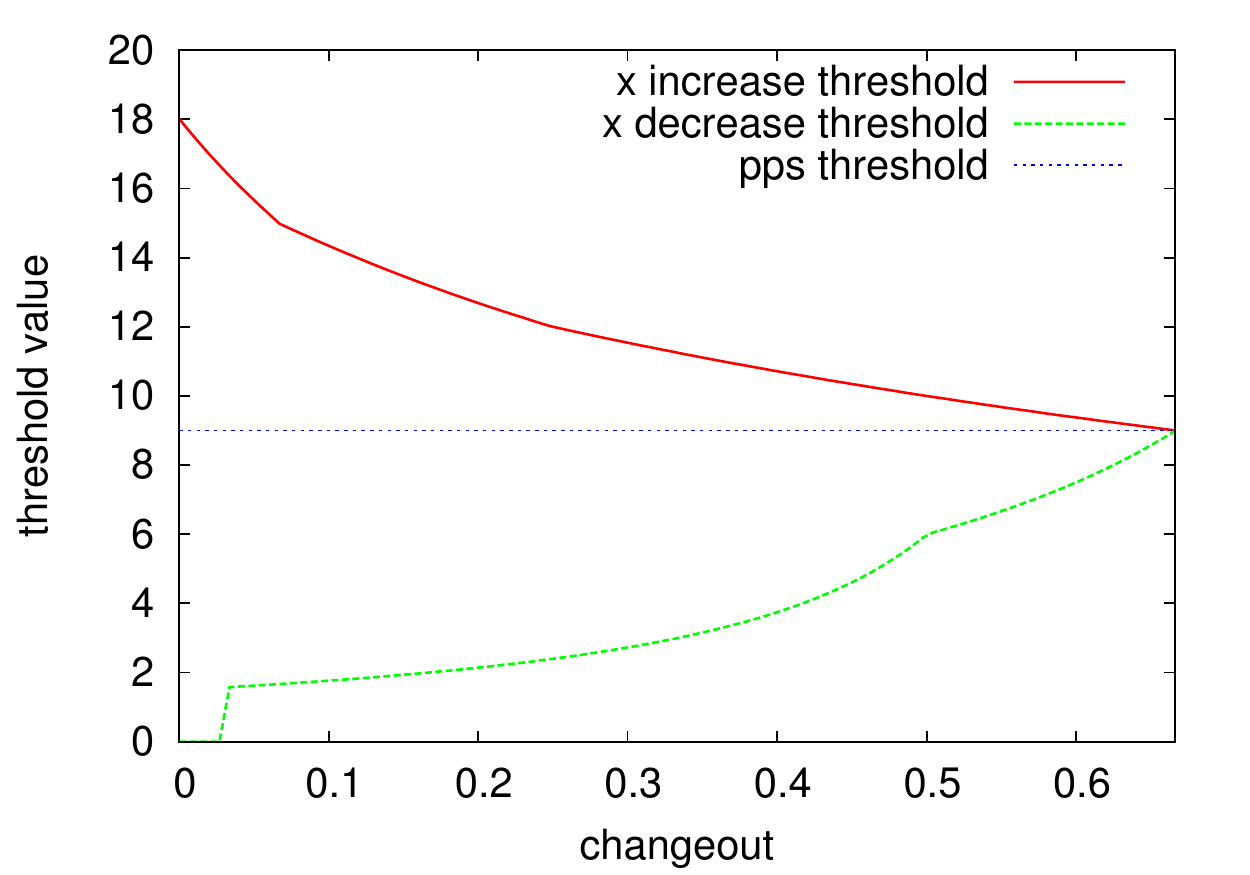} 
\caption{The thresholds $\underline{\tau}$ and $\overline{\tau}$ as a
  function of changeout $x$ for Example~\ref{ex1}.  PPS threshold with
  maximum changeout ($4/3$) is provided as a reference.\label{threshex:assign}}
\end{figure}

\begin{lemma}
A representation of  $\Delta\text{-{\sc
    Opt}}(\boldsymbol{w},\boldsymbol{p},D)$ 
 can be computed 
in $O(n\log n)$ time.  The representation 
supports  computing
the output distribution for a particular value $D$ in time linear in the
number of entries with $p_i\not= q_i$.
\end{lemma}
\begin{proof}
Algorithm~\ref{optdelta} computes 
$\Delta\text{-{\sc Opt}}(\boldsymbol{w},\boldsymbol{p},D)$ in $O(n)$ time after computing
the  thresholds $\overline{\tau}(x)$ and $\underline{\tau}(x)$ for 
$x=D/2$, using
Algorithms~\ref{deltaplus} and~\ref{deltaminus}.
If we are interested in a particular $D$, the algorithms can be stopped once the piece of the function
containing $x=D/2$ is computed.  This takes $O(n\log n)$ time of Algorithms~\ref{deltaplus} and~\ref{deltaminus}.
The complete representation is the
combination of the representations of $\underline{\tau}(x)$ and
$\overline{\tau}(x)$.  If it is precomputed, Algorithm~\ref{optdelta}
returns the distribution in time linear in the number of modified entries.
\end{proof}

\begin{algorithm}[t]  \caption{Compute the distribution $\Delta$\text{-{\sc
        Opt}}$(\boldsymbol{w},\boldsymbol{p},D)$ \label{optdelta}}
 \begin{algorithmic}
 \Function {Opt}{$\boldsymbol{w},\boldsymbol{p},D$}
\State $x \gets D/2$
\State $\underline{T} \gets \underline{\tau}(x)$ \label{lineutau}
\State $\overline{T} \gets \overline{\tau}(x)$ \label{lineotau}
\If {$\overline{T}=0$} $R\gets x$ \EndIf
\For {$i\in [n]$}
\If {$p_i\leq  \frac{w_i}{\underline{T}}$}
\State  $q_i \gets  \frac{w_i}{\underline{T}}$
\ElsIf   {$\overline{T}=0$}
\If {$R>0$ and $w_i=0$}
\State $q_i \gets \max\{0,p_i-R\}$
\State $R \gets R-p_i$
\EndIf
\ElsIf {$p_i <  \frac{w_i}{\overline{T}}$} \Comment{$\overline{T}>0$ and
  $p_i \in [\frac{w_i}{\overline{T}},
  \frac{w_i}{\underline{T}})$}
\State $q_i \gets  \frac{w_i}{\overline{T}}$
\Else \Comment{$\overline{T}>0$, 
  $p_i \in [\frac{w_i}{\overline{T}}, \frac{w_i}{\underline{T}})$}
\State $q_i \gets p_i$
\EndIf
 \EndFor
\\\Return $\boldsymbol{q}$
 \EndFunction
\end{algorithmic}
\end{algorithm}

\subsubsection*{Computing an $\alpha$-stable PPS distribution}

 For a given stable solution, we can consider its changeout $D$ from
 the starting distribution and the Lagrange multiplier $a$, which
is  the rate of decrease of the sum of variances per unit change after implementing a
total change of $D$ in sampling probabilities.  The relation between
the two is
$$\alpha(D)\equiv a =\underline{\tau}(D/2)^2-\overline{\tau}(D/2)^2
.$$ 

\begin{lemma}
A description of the function $\alpha(D)$, and a solution for an equation
$\alpha(D)=a$,  can be computed in $O(n\log n)$ time.
\end{lemma}

\begin{proof}
The function $\alpha(D)$ is a piecewise rational expression which is decreasing
with $D$.  Its breakpoints are the union of the breakpoints of
$\underline{\tau}$ and $\overline{\tau}$.  
For each piece, we can compute the respective range of $D$ values and
of $a$ values.
To compute the
$\alpha$-stable  distribution for a given $a$, we locate the piece in
which $a$ is in the range and solve for $x$ the
equation $a=\underline{\tau}(x)^2-\overline{\tau}(x)^2 .$
Within a specified piece, this yields a quartic equation
(finding the root of a polynomial of degree $4$), and we take the
root $x$ that lies in the domain of the piece.  

Once $x$ is computed, we have the changeout
$D=2x$ which corresponds to $a$ and
the thresholds $\underline{\tau}(x)$ and $\overline{\tau}(x)$.  We can 
substitute the thresholds in the initialization of 
Algorithm~\ref{optdelta} and apply
it to compute the $\alpha$-stable distribution. 
\end{proof}

Combining these pieces provides the result claimed by
Theorem~\ref{thm:pps}.

\subsection{Incrementally maintaining the stable PPS distribution} \label{dynpps}
We now show how we can
efficiently maintain an $\alpha$-stable distribution when
  modifications to the weight vector are made one entry at a time.
Consider an $\alpha$-stable distribution and the set of ratios
$w_i/p_i$.  
Following the notation we used in the batch case, 
let $\underline{\tau}$ be the largest ratio with
  $p_i<1$; and let $\overline{\tau}$ be $0$ if there is an entry with
zero weight and positive probability and otherwise the smallest ratio
with $w_i>0$.  The difference of the squares is at most $a$:
$\underline{\tau}^2-\overline{\tau}^2 \leq a$.  When the weight of an
entry $i$ is modified, its ratio changes.    
If this results in a maximum difference of
squares exceeding $a$,  we need to recompute the $\alpha$-stable distribution.
This event can happen only when  either
$p_i<1$ and $w_i/p_i > \underline{\tau}$ or
$p_i=0$, $w_i>0$ and $\overline{\tau}>0$ or $w_i/p_i <
\overline{\tau}$.
  The modification that involves the modified entry itself  must be at
  one end of the ordered ratios list and a sequence of entries from the other end of the ordered ratios.
We increase $\overline{\tau}$ to simultaneously compensate for decreasing
$\underline{\tau}$ or vice versa, but noticing we never need to exceed the
original threshold value at the affected end.
 The total number of modified probabilities, however, can be large,
 even when we modify the weight of a single entry.  On the other hand, the
expectation of the change in the sample (the $L_1$ distance between
the initial and final $\alpha$-stable distributions) is at most $1$
and we can ensure (as we outlined in the subsampling discussion) that there
is at most one insertion and one deletion from the sample.

Our treatment of the incremental problem introduces a
 representation of the $\alpha$-stable distribution, which allows for
  efficient tracking and modifications of  sampling probabilities
  of sets of entries and also efficient maintenance of a corresponding
  sample.
This is combined together with amortized  analysis which charges the
modifications to the representation to previous modifications in the input.
We show the following  (the details are deferred to Appendix~\ref{moredynpps}):

\begin{theorem} \label{dynamicpps}
An $\alpha$-stable PPS distribution, and a corresponding sample, can be maintained in amortized
$\operatorname{poly}\log(n)$ time per modification.
\end{theorem}

\eat{
\subsubsection*{Fixed size sample ?}

  Varopt sample does not have the equivalent of  {\sc Subsampling}.
  Simple example show that in general, it is not possible to subsample
  optimally and retain varopt properties.

  One possible approach with bottom-$k$ is to use fixed seeds to
  compute ranks (priorities) and apply the top-$k$ like swaps to the priorities.
Effectively, this gives two thresholds for each step.  One for staying
in and one for getting admitted/ejected.

 Using the thresholds and the data in previous states, we can compute
 effective inclusion probability for each item (fixing seeds of other
 items) and obtain unbiased adjusted weights.
}

\section{Additive fitness and changeout cost}

We now describe a general approach for a broad class of optimization
problems.  
We say an optimization problem with input domain of
the vectors  ${\cal X} \subset R^n$ and outputs
  ${\cal S}$  that are subsets $S\subset [n]$ has {\em additive
    fitness} if its
  objective is to maximize the fitness function $\phi(S,\boldsymbol{x})=\sum_{i\in S} x_i$:
 $$\text{{\sc opt}}(\boldsymbol{x})=
{\arg\max_{S\in {\cal S}} \sum_{i\in S}} x_i\ .$$
The changeout cost between two outputs is {\em additive}
if there are constants $c_i$ such that
$d(S,S')=  \sum_{i\in S'\setminus S} c_i$.  
The constant $c_i$ can be interpreted as
the cost of ``bringing in'' $i$ to the output.
When the optimization problem has additive fitness and changeout cost,
the  $\alpha$-stable optimum w.r.t. current output $S$ can be
posed as a best-fit instance  on a modified input
vector with an appropriate quantity subtracted from entries in $i\not\in S$:

\begin{theorem}  \label{alphaoptreduction}
Let {\sc opt} be an optimization problem with additive fitness and changeout.
The $\alpha$-stable optimum 
satisfies
$\alpha\text{-{\sc opt}}(S,\boldsymbol{x},a)= \text{{\sc opt}}(\boldsymbol{y})\
,$ where $y_i=x_i - I_{i\not \in S}  c_i a$ and $I$ is the indicator function.
\end{theorem}

\begin{proof}
From \eqref{localopt},
\begin{align*}
\alpha\text{-{\sc opt}}(S,\boldsymbol{x},a)  & = {\arg\max_{S'\in {\cal S}}}
 \Big(- a {\sum_{i \in S'\setminus S}} c_i + {\sum_{i\in S'}} x_i \Big)\\
&= \arg\max_{S'\in {\cal
    S}} {\sum_{i\in S'}} y_i = \text{{\sc opt}}(\boldsymbol{y})\ .
\end{align*}
\end{proof}

We now consider maintaining an $\alpha$-stable optimum efficiently
under incremental updates. 
We show that  this can be done using  an off-the-shelf dynamic {\sc opt} algorithm:
\begin{theorem}  \label{dynamic}
If fitness and changeout are additive, then
when the input is modified,
the $\alpha$-stable solution can be maintained by applying a dynamic
{\sc opt} algorithm with a corresponding number of modified entries. 
\end{theorem}

\begin{proof}
For an input $\boldsymbol{x}$ and current output ${\bf S}$, we consider the
vector $\boldsymbol{y}$ such that
 $y_i=x_i- I_{i\not \in S} a c_i$.
 Recall that the $\alpha$-stable
optimum is the best-fit with respect to input $\boldsymbol{y}$.  It therefore
suffices to translate the modifications 
in $\boldsymbol{x}$ to the modifications in $\boldsymbol{y}$.
 When an entry of $\boldsymbol{x}$
is modified, we apply a  corresponding modification to the
same entry in $\boldsymbol{y}$.  When the algorithm modifies the output, we
increase  by $a c_i$ the value of new entries inserted to $S$ and
decrease by $a c_i$ the value of entries that are deleted from $S$.
The total number of updated entries in $\boldsymbol{y}$ is the sum of  the
updates to the input and the set difference between the $\alpha$-stable outputs.
\end{proof}

 \subsubsection*{ Fixed-size output and uniform costs}
In this section, we study problems that possess
additional structure, namely {\em fixed-size output},
$\forall S\in {\cal S}, |S|=k$ and {\em uniform costs} $c_i=1$.
In particular, this formulation applies to matroids.
We can equivalently treat changeout costs as those of excluding
entries currently in the output.  
Examples of such problems include Top-$k$, MST, and assignment, which
we study individually below. 

Specializing
Theorems~\ref{alphaoptreduction} and \ref{dynamic},
the $\alpha$-stable optimum can be computed by applying {\sc opt} to a
vector which adds a fixed value $a$ to entries in $S$:
$\alpha\text{-{\sc opt}}(S,\boldsymbol{x},a)= \text{{\sc opt}}(\boldsymbol{y})\
,$ where $y_i=x_i + I_{i\in S} a$:
\begin{align*}
\alpha\text{-{\sc opt}}(S,\boldsymbol{x},a) &= \arg\max_{S'\in {\cal S}}
 \Big(- a |S'\setminus S| + \sum_{i\in S'} x_i \Big) \\
 &= \arg\max_{S'\in {\cal S}}
\Big( a |S'\cap S| + \sum_{i\in S'} x_i \Big) \\
&=\arg\max_{S'\in {\cal
    S}} \sum_{i\in S'} y_i = \text{{\sc opt}}(\boldsymbol{y})\ .
\end{align*}
Since $\boldsymbol{y} \geq \boldsymbol{x}$,  our reduction can also be used when the domain of
{\sc opt} is non-negative.

\noindent
We can now understand the structure of the
objective of $\alpha$-{\sc opt} as a 
 function of $a$:
\begin{lemma}
\label{lem:linear}
The function $\max_{S'\in {\cal S}}
\big(\sum_{i\in S'} x_i +a |S\cap S'| \big)$
is the maximum of at most
$k+1$ linear functions.  Thus, it is convex and piecewise linear with at most $k$ breakpoints.
\end{lemma}
\begin{proof}
The objective is a  (restricted)
  {\em parametric} version of the original optimization
  problem~\cite{Mi,Karp,KarpOrlin81,GGT,Gusfield}.
In a solution with respect to a parameter $a$, the
weight of each entry is replaced by the linear function
$x_i+a$ when $i\in S$ and is fixed to $x_i$ otherwise.  
The value of a particular output $S'$ with respect to $S$ and $a$ is
$\sum_{i\in S'} x_i + a|S\cap S'|$.  
The optimal solution as a function of $a$ 
is the maximum over all $S'\in {\cal S}$.  
These are linear functions with at most $k+1$ different slopes,
corresponding
to the intersection size $|S\cap S'|$.  For each intersection size
$\{0,\ldots,k\}$,
there is one dominant linear function.   The maximum of linear
functions
with slopes $\{0,\ldots,k\}$ is such that each function 
dominates in at most one piece and the slopes are increasing.
\end{proof}

\noindent
If we explicitly compute a representation of the objective function,
we can obtain a $\Delta$-stable optimum, by identifying
a corresponding value of the parameter $a$:
The breakpoints of the function correspond to decreasing
changeouts (increasing slopes) from changeout $k$ (slope $0$) to
changeout $0$ (slope $k$).
The marginal
gain in fitness from each unit increase in changeout is the
respective value of
$a$ at the breakpoint.  Hence: 
\begin{corollary}
When the problem has additive fitness,  additive and uniform changeout, 
and fixed-size
output, 
the tradeoff is concave. 
\end{corollary}

\noindent
There are at most $k+1$ different $\alpha$-stable solutions (eliminating
duplications) and thus {\em all} can be specified in $O(k^2)$ space.
When the tradeoff is monotone, then this description requires only $O(k)$ space.   This is because there are
at most $2k$ entries involved in total, and each has a single exit or
entry point as we sweep the parameter.  
We will show that monotonicity holds for Top-$k$ and MST but not for assignment.

\subsection{Stable top-$k$}

The top-$k$ problem has input domain ${\cal X} \subset R^n$ so that 
$\boldsymbol{x}_i$ is the `weight' of $i$. 
The permitted outputs ${\cal S}$ are all subsets $S\subset [n]$ of
size $k$, and the goal is to maximize $\phi$, the sum of weights of
the selected subset.  
The problem satisfies the conditions of Theorem~\ref{alphaoptreduction}
and therefore the  $\alpha$-stable top-$k$ with respect to input $\boldsymbol{x}$ and current set $S$  is
 $\alpha\text{-}\text{top-}k( S,\boldsymbol{x},a) \equiv \text{top-}k(\boldsymbol{y})$, where $y_i=x_i+ I_{i\in S} a$.  


\begin{lemma}
The stability-fit tradeoff for stable top-$k$ is monotone
and can be found in $O(n+k\log k)$ time.
\end{lemma}

\begin{proof}
Consider the sequence of swaps made by
${\alpha}\text{-top-k}( S,\boldsymbol{x},a)$ while sweeping $a$.
For sufficiently small $a \geq 0$,
  ${\alpha}\text{-top-k}( S,\boldsymbol{x},a)=S$.  
For sufficiently large $a$,
${\alpha}\text{-top-k}( S,\boldsymbol{x},a)=\text{top-k}(\boldsymbol{x})$.
The total number of swaps is
the set difference $d(S,\text{top-k}(\boldsymbol{x}))$. 
The swaps are between the lightest $i\in S$ and
the heaviest $i\not\in S$.

To compute all tradeoff points, we sort the entries in $S$ by
increasing $x_i$ to obtain the order
$(i_1,\ldots,i_k)$, 
and sort the top-$k$ entries $i\not\in S$ in decreasing order
 $(j_1,\ldots,j_k)$.
This obtains the stated time bounds. 
Let $H=\arg\max_h x_{i_h}< x_{j_h}$.
For $h\leq H$,  the most beneficial swap of at most $h$ items,
swaps-in  items $j_1,$ $\ldots,j_h$ and swaps-out items $i_1,\ldots,i_h$.
The $\alpha$-stable swap uses the maximum $h\leq H$
such that $x_{j_h}-x_{i_h}\geq a$.
The tradeoff is clearly monotone, as each entry is swapped in or out
 at one particular value.
\end{proof}

\noindent
{\bf Incremental algorithm:}
The problem also satisfies the conditions of Theorem~\ref{dynamic}
and therefore to obtain an incremental algorithm, we outline 
a simple dynamic algorithm for maintaining a $\text{top-}k$ set.
We use two heap data structures.
The first heap is a min-heap of the $k$ cached entries
(we call it the ``in-heap'').
The second heap is a max-heap of the remaining $n-k$ entries (``out-heap'').
Weight updates are performed as updates to the weights of
the corresponding elements in the heap. 
Swaps cause the
deletion of the minimum item in the in-heap and the maximum item in
the out-heap, and their reinsertion in the other heap.
Swaps are performed when $x_i-x_j\geq \alpha$, where
$i$ is the maximum item in the out-heap and $j$ is the minimum item in
the in-heap.  
In this algorithm each update has an overhead of $O(\log n)$
---
note that an update of a single value can result in at most a single swap.

\medskip
\noindent
 {\bf The optimal offline algorithm:}
 We now provide an optimal algorithm for offline stable top-$k$.
The offline problem is to maximize
$$\sum_j\sum_{i\in S_j} x_i^{(j)}-a \sum_j |S_j\setminus S_{j-1}|\ .$$
We can obtain the offline stability-fit  tradeoff by sweeping $a$.
The same set of solutions can be obtained by sweeping the maximum
number of cache swaps $K$,   and  maximizing
$\sum_j\sum_{i \in S_j} x_i^{(j)}$ subject to 
$\sum_j |S_j\setminus S_{j-1}|\leq K$.
For each item $i$ and subsequence $[j,\ell]\subset [N]$, we compute
the benefit of caching $i$ during the subsequence:  
$$b_{ij\ell}=\sum_{h=j}^\ell x^{(h)}_i\ .$$  
We seek a cover with at most $K+k$ distinct intervals such that there
are at most $k$ active intervals at each step $h$ and the cover has maximum benefit.
 This can be formulated as an integer program with variables
  $y_{ij\ell}\in[0,1]$:
\begin{align*}
& \text{maximize } \sum_{ij\ell} b_{ij\ell}  y_{ij\ell}\\
& \text{subject to }\sum_{ijl} y_{ij\ell} \leq K+k  \\
\text{and } \forall h\in [n],\  & \sum_{ j\leq h,} \sum_{\ell\geq h} \sum_{i\in
  [n]} y_{ij\ell} \leq k 
\end{align*}

\eat{
\[\text{maximize } \sum_{ij\ell} b_{ij\ell}  y_{ij\ell},
\quad 
\text{subject to}
\quad
\sum_{ijl} y_{ij\ell} \leq K+k  
\quad
\text{and}
\quad
\forall h\in [n],\   \sum_{ j\leq h,} \sum_{\ell\geq h} \sum_{i\in
  [n]} y_{ij\ell} \leq k 
\]
}

This problem can also be modeled as min-cost multi-commodity flow. 
It can also be solved via dynamic programming in polynomial time.
For maximum stability ($K=0$), we simply choose the top-$k$ keys
with maximum $b_{i,1,N}$.

\medskip
\noindent
{\bf Example:}
Consider an input vector $\boldsymbol{x} = \langle 1, 4, 7, 5\rangle$.
 The top-2 set is the entries $\{3,4\}$.
Suppose $k=2$ and at the current output is the first two entries
$S=\{1,2\}$. 
The best  swap is $1 \leftrightarrow 3$
(swapping the smallest cached value with the largest uncached value) and
the second most beneficial swap is $2 \leftrightarrow 4$ (swapping 2nd
smallest cached with 2nd largest uncached).  The optimal
tradeoff had $3$ points:  performing no swap, only the first swap, or
both swaps.  
Looking at $\alpha$-stable solution,
the first swap has utility gain $7-1=6$ and the second swap has utility gain 
$5-4=1$.  Thus the first swap is performed when $\alpha\leq
6$ and both when $\alpha\leq 1$.  
\qed

\medskip
 Our results carry over to
fitness functions with the more general  form
$\phi(S,\boldsymbol{x})=\sum_{i\in S} \psi_i(x_i),$  by treating $\psi(\boldsymbol{x})$ as the input.
This generalization allows us to 
measure fitness by any  $L_p$ norm, using $\psi(x)=|x|^p$.
The choice of $p$ matters when processing a sequence of inputs using
the same $a$ value:  Augmenting our example,
consider the $\alpha$-stable outputs for two inputs:
$\boldsymbol{x}= \langle 1, 4, 7, 5\rangle$ and  another vector
$\boldsymbol{z}= \langle 2, 3, 8, 4 \rangle$, both starting with current output
$\{1,2\}$.  The vector $\boldsymbol{z}$ has the same tradeoff curve as $\boldsymbol{x}$ when using the $L_1$ norm.
If we want to use the $L_2$ norm for
fitness,  we use $\psi(x)=x^2$.  In this case, the two most valuable swaps have utility gains of
$7^2-1^2=48$ and $5^2-4^2=9$ on $\boldsymbol{x}$ and
$8^2-2^2=60$ and $4^2-3^2=7$ on $\boldsymbol{z}$.  Therefore, on the same
$\alpha$, and looking at both vectors together, it is possible that only the first swap is performed on the
$\boldsymbol{x}$  ($60 \geq \alpha> 48$), only the first swaps are performed
on both ($48 \geq \alpha> 9$),  the two swaps 
on $\boldsymbol{x}$ and only the first swap on $\boldsymbol{z}$
($9 \geq \alpha> 7$), or all four swaps ($\alpha \leq 7$).  

\subsection{Stable MST}

In the stable MST problem, the inputs are a set of
edge weights  and the outputs are spanning trees.  The fit of a tree
is its negated weight.  The change cost between two trees is the set
difference of included edges.
The problem clearly satisfies the conditions of
Theorem~\ref{alphaoptreduction},
hence, an $\alpha$-stable MST can be obtained by adding $a$ to the edge weights
  of edges in the current output tree and computing an MST.
Applying Theorem~\ref{dynamic},
dynamic $\alpha$-stable  MST can be handled using off-the-shelf dynamic MST 
algorithms~\cite{HolmLT:JAMC01}.

 \begin{lemma}
The tradeoff for stable MST is monotone.
\end{lemma}
\begin{proof}
 Consider starting
from an optimal MST and slowly increasing $a$, obtaining a sequence of
$\alpha$-stable MSTs.  An edge is removed 
when it is the heaviest edge in some cycle.  This can only happen if
this edge is not a member of the output $S$ (because weights of all
members of $S$ decrease by the same amount and all weights of other
edges stay the same).  Once this happen, weights of other edges in the
cycle either keep decreasing or stay the same, so the edge remains the
heaviest in that cycle and out of the MST.  Similarly, an edge is
inserted into the $\alpha$-stable MST when it has smallest weight in
some cut. Such an edge must be a member of $S$ and its weight remains
the smallest in its cut as $a$ increases.  Thus the tradeoff is
monotone.  
\end{proof}
The tradeoff can be computed using a parametric MST
algorithm~\cite{parMST:focs98}, but the computation can be 
made more efficient for this special case.

\begin{figure}[t]
\centering
 \includegraphics[width=0.3\textwidth]{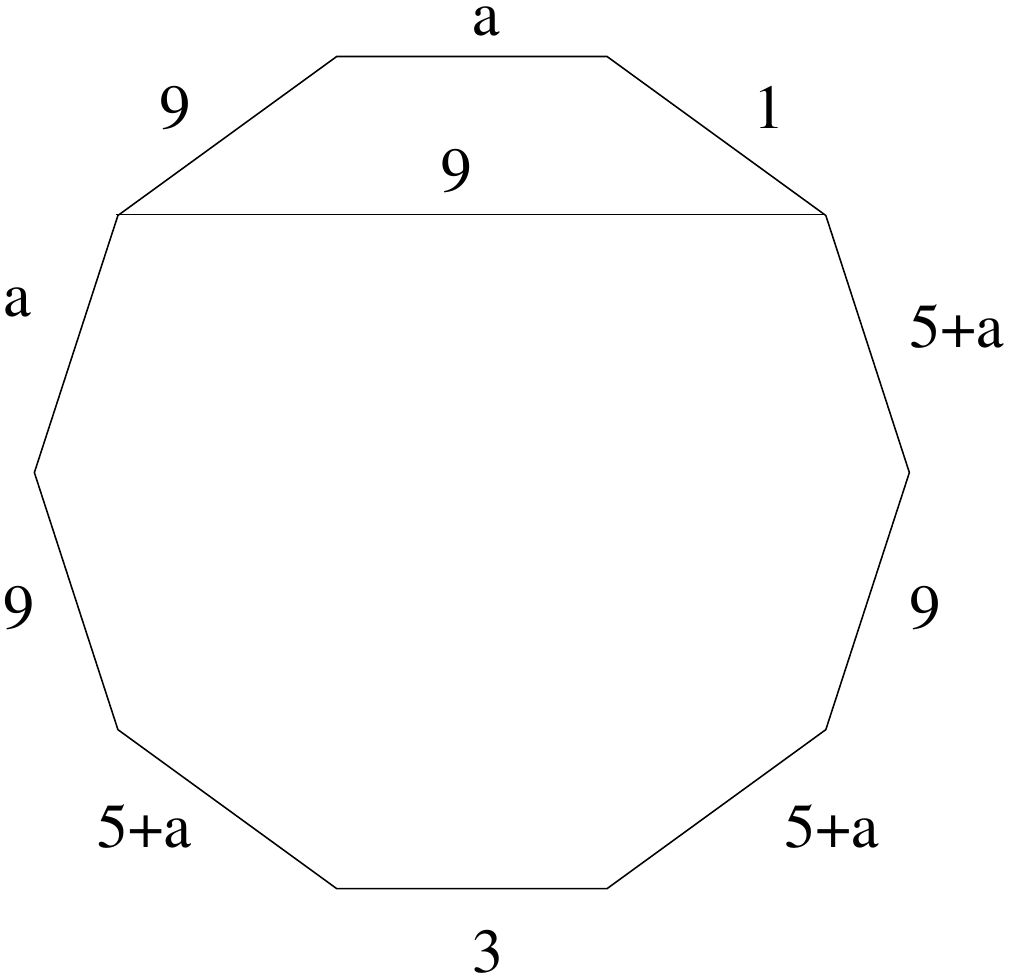} 
\caption{Example demonstrating non-monotonicity of assignment}
\label{fig:assign}
\end{figure}

\subsection{Stable assignment}

In the stable assignment problem, the inputs are weights of edges in
  a complete  bipartite graph and the outputs are maximum matchings.
The problem satisfies the conditions of
Theorem~\ref{alphaoptreduction},  and hence, 
an $\alpha$-stable solution can be computed by adding $a$
to all the weights of edges present in the current matching and
  then finding a maximum-weight assignment  (maximum weight bipartite
  matching).
We also know there are at most $n+1$ distinct stable optimal
assignments (assuming a single optimal assignment is used).
From Theorem~\ref{dynamic},
dynamic $\alpha$-stable  assignment reduces to dynamic maximum
bipartite matching~\cite{Sankowski:soda07}.

\begin{figure*}[t]
\centering
\begin{tabular}{cc}
\subfigure[Stable PPS uniformly and noticeably outperforms
  EWMA PPS. Permanent Random Numbers (PRN) improve both
  methods.]{
\includegraphics[width=0.45\textwidth]{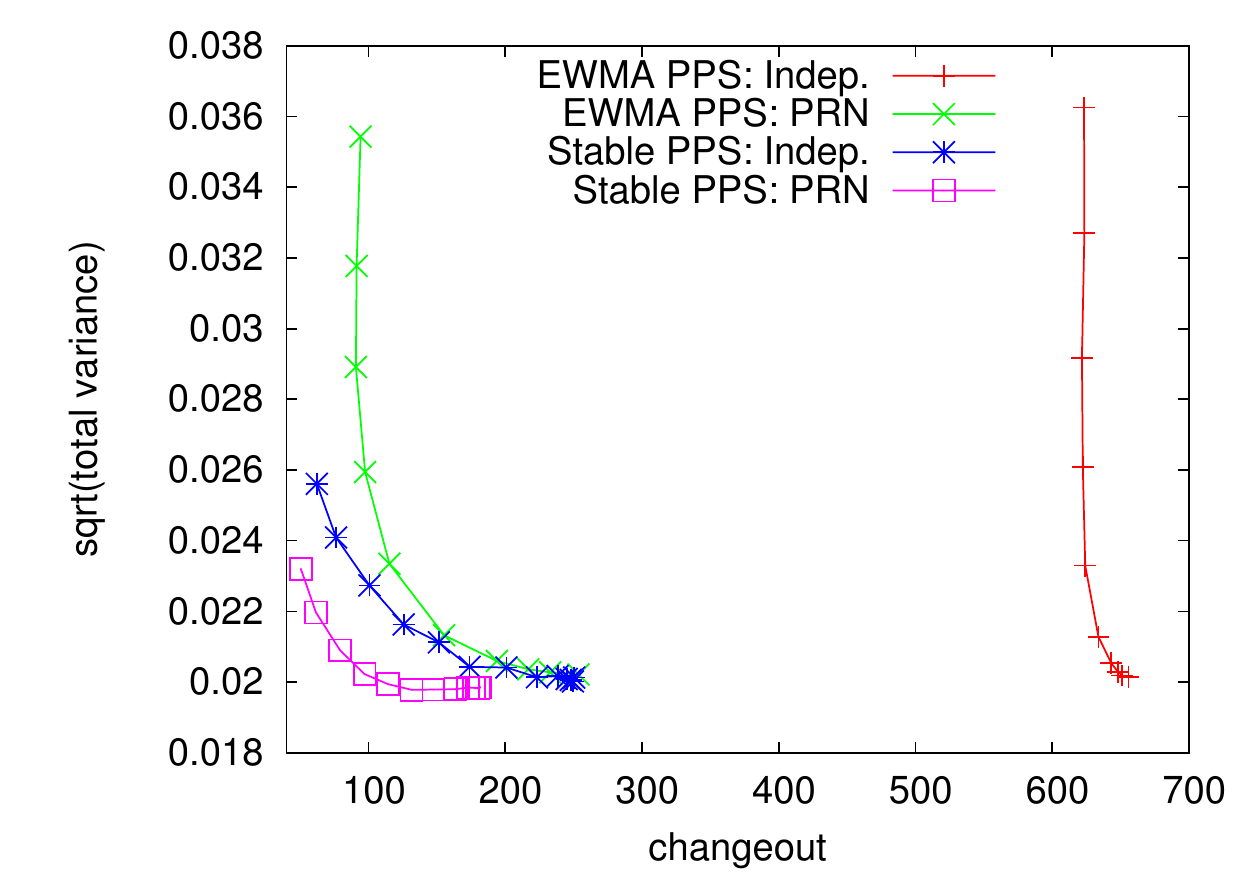}
\label{thefigsa}
}
&
\subfigure[ Similar performance for Stable and EWMA top-$k$,
  although computational aspects favor Stable top-$k$.]{
\includegraphics[width=0.45\textwidth]{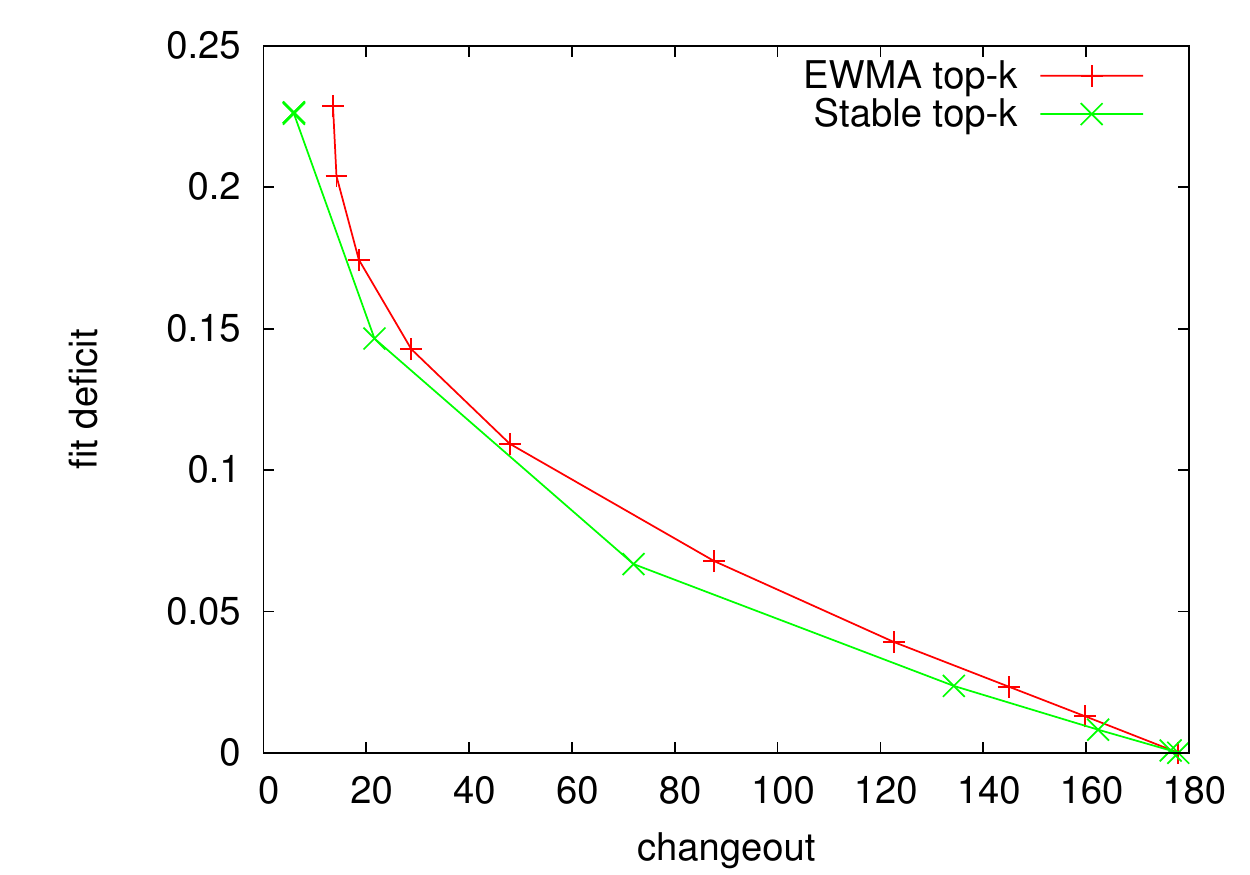}
\label{thefigsb}
}
\\
\end{tabular}
\caption{Tradeoff between stability and fit comparing Stable with Time
  Decay EWMA. }\label{thefigs}
\end{figure*}

However, unlike the previous problems, the tradeoff is non-monotone.
%
Specifically, 
Figure~\ref{fig:assign} shows a simple example of an assignment
problem which is non-monotone in its inclusion of edges. 
It shows a 10 node bipartite weighted graphs (where missing edges are
assumed to have weight zero). 
A matching from the previous timestep causes five edges to have their
(current) scores increased by $a$. 
We can find the tradeoff by considering solutions which include
different numbers of the previous edges, as per Lemma~\ref{lem:linear}. 
The solution that picks none of the previous edges has weight 31.
The solution that picks one of the previous edges picks the top edge,
and has weight $30+a$. 
The solution that picks three previous edges has weight $24+3a$, 
and picks the three edges with weight $5+a$. 
Lastly, the solution that picks all five previous edges has weight
$15+5a$. 
The cases of two and four previous edges are dominated by these
solutions. 
The $\alpha$-{\sc opt} objective is the maximum of the four linear equations:
$y=31$, $y=30 + 1a$, $y=24 + 3a$, $y=15 + 5a$, and each
dominates the other three for a different range of a.
Further, note that the top $a$ edge is included for $a\in (1,3)$ and
not for $a\in (0,1)$ and $(3,4.5)$.  
Therefore, the tradeoff is not monotone in its inclusion of edges.

\section{Illustration and Comparison with Time Decay}

We provide a brief performance comparison of Stable PPS and Stable
top-$k$ against their time-decayed counterparts.

\medskip
\noindent
{\bf Experimental Set-up.}
The data comprises IP flow records aggregated over the 144 ten-minute
periods over the span of a day, obtained from a large ISP. 
On average, there were $61,292$ (anonymized) keys in each period $t$,
from a total of $2,516,594$. 
The results shown are for a summary size of $k=1,000$ in each case. 

For the time decayed version we constructed for each key $i$ an
Exponentially Weighted Moving Average (EWMA) $\tilde x_{i,t}$ of the
process $x_{i,t}$. 
In EWMA PPS, at each time $t$, we perform PPS
selection over the keys present but using EWMA weights, i.e., the
weight set $\{\tilde x_{i,t}: x_{i,t}>0\}$, and formed the
Horvitz-Thompson estimate of the true weight as an auxiliary
variable, i.e., $x_{i,t}/\tilde p_{i,t}$ for selected items, where
$\tilde p_{i,t}$ is the selection probability.
In EWMA Top-$k$ we selected at each time 
$t$ the set $\tilde S_t$ of $k$ items of highest EWMA weight 
$\tilde x_{i,t}$, and as measure of fit computed the corresponding 
sum of original weights $\sum_{i\in \tilde S_t}x_{i,t}$.

\medskip
\noindent
{\bf PPS Sampling:} Figure~\ref{thefigsa} shows
the stability-fit tradeoffs for Stable PPS and EWMA PPS.
Fit is measured as the square root of sum of weight estimation variances.
Fit and changeout costs are averaged over the 144 periods. 
For EWMA PPS we swept the mean $m$ of the geometric decay from $1$ (no
history) to $64$ periods. 
We computed the $\Delta$-stable PPS distributions sweeping
target maximal changeout cost $D$ between $0$ and the target sample
size $k=1,000$. 
For both methods we compared independent (pseudo)random number
generation, and permanent random numbers (PRN) for each key. 

EWMA PPS with independent random numbers has much larger
changeout cost (in the range 620--650) than the other methods. 
Using key based PRN with EWMA PPS much reduces the change cost. 
Stable PPS performs noticeably
better, achieving very small changeout (around 50 for the PRN variant) 
with relatively little growth in error as changeout is
reduced (about 15\% growth compared with about 70\% growth for EWMA).
PRN has little impact on variance because, for a given
input set, marginal variance is unchanged.

\medskip
\noindent
{\bf Top-$k$:}
Figure~\ref{thefigsb} shows stability-fit tradeoffs for
Stable top-$k$ and EWMA top-$k$, where fit is represented by
the deficit between the total weight of the top-$k$ estimate, and the
true total top-$k$ weight.
Fit and changeout are again averaged over the 144 periods. 
The performance of both are similar, although the stable version
does obtain a better tradeoff. 
Nevertheless, we note that stable methods are
preferable to EWMA methods for computational reasons: they are generally
faster, and do not require historical data to be retained to determine
the output.
Moreover, the stability or flexibility parameters can be
modified directly at each iteration, whereas under time-decay
these parameters can be controlled only indirectly through the decay
parameters.
Changing decay parameters 
requires extensive historical data to be retained, and rescanned for
each change.

\section{Stable Extensions of Optimization Problems}

In general, we can consider the stable extension of almost any
optimization problem.
Theorem~\ref{alphaoptreduction} shows that for a natural class of
problems, $\alpha$-stability reduces to a single modified instance of
the original optimization problem.  
However, there are many other problems which do not admit such a
reduction. 
We considered the case of PPS sampling, and showed procedures which
find the optimal solution. 
Also of interest are those problems where the original
optimization is known to be hard: then we seek approximation
algorithms for the stable extension (making use of
Theorem~\ref{alphaoptreduction} if possible). 

In this section, we provide further examples of problems that
naturally admit stable extensions: 
\begin{trivlist}
\item
In the {\em 0/1 knapsack problem}, items have values $x_i$ and weights
$w_i$, and we
have a parameter $W$.  The goal is to find a set of items of
maximum value, subject to an upper bound of $W$ on  the
sum of their weights, $\sum_{i \in S} w_i \leq W$.
Under additive changeout cost, 
this fits the requirement of Theorem~\ref{alphaoptreduction}, 
and so any approximation algorithm for 0/1 knapsack provides the
corresponding approximation for the $\alpha$-stable version of the
problem. 
\item
In {\em shortest-paths routing} the output is a tree routing to a single
destination.  The fitness is the (weighted) sum, over nodes, of the
path length to the root.  The changeout cost between two trees is
the set difference (viewing trees as directed to the destination),
which corresponds to the number of modifications to the routing table for
the destination.
\item
The {\em $k$-set cover problem} is,  given a fixed collection of sets, to
pick at most $k$ to maximize coverage of $x$.
The fitness is the sum of weights of covered items, and we consider
the uniform changeout cost function.  We outline a reduction
from an  $\alpha$-stable $k$-cover to the optimum $k$-cover on a modified
input.
Our modified instance contains all the original items and sets.  It
also contains a new item that is unique for each set in
the current output  $S$.  The new items have weights $a$.
We argue that a cover is $\alpha$-stable for $a$ if and only if it is
an optimal $k$-cover of the modified input.
\end{trivlist}

\subsection{$k$-center Clustering}
\label{app:cluster}

It is natural to consider the stable extension of various {\em
  clustering} problems. 
Here, we wish to minimize the cost of the new clustering plus the cost
of switching out some of the cluster centers. 
In this section, we show how to find a constant factor
approximation for the $k$-center objective in this setting. 

Given the stable extension of any optimization problem, 
an $\alpha$-optimal solution can always be computed by exhaustive
search on possible outputs.
A method with smaller  search space is to first construct an algorithm for
an extended version of the optimization problem
which allows a part of the output to be fixed.  
This extension is a special case of the stable extension. 
We can then perform a smaller
search over all subsets of the current output,
and apply the (restricted) extended algorithm to each instance.  
Some problems naturally
yield to such extensions:  for covering problems, we remove items
covered by the fixed component and apply the (approximate) covering
algorithm to what remains.

A less straightforward case arises in the context of clustering. 
Here, we want to extend a given subset of cluster centers with some
additional centers to minimize a particular cluster objective.  
For the $k$-center objective, the existing
2-approximation algorithm can be extended to handle fixed facilities
(this ``restricted'' extension is also interesting in itself).
More formally, for points in a metric space $m$, we have:

\begin{lemma}
Given a fixed set of centers $F$ and a set of $n$ points $P$, we can find
a set of clusters $C$ of size $s$ that guarantees to 2-approximate
the $k$-center objective, i.e.
\[ \max_{p \in P} \min_{c \in C \cup F} m(p,c) \leq 
   2 \min_{C : |C| = s} \max_{p \in P} \min_{c \in C \cup F} m(p,c) \]
\end{lemma}

\begin{proof}
Here, we adapt the 2-approximation algorithm 
due to Gonzalez~\cite{Gonzalez:85}. 
Starting with $C = \emptyset$, we pick the point from the input that
is furthest from the current centers in $C \cup F$, i.e. 
\[ C \gets C \cup \{\arg\max_{p \in P} \min_{c \in C \cup F} m(p,c)\} \]
and iterate until $|C|=s$. 
This requires $O((|F|+s)n)$ distance computations. 

To see the approximation guarantee, consider the next point $q$ that would
be added to $S$ if we continued the algorithm for one more iteration. 
This is some distance $d$ from its closest center in $C \cup F$. 
Therefore, we have that all points in $C \cup \{q\}$ are separated by
$d$, and further that every point in $C \cup \{q\}$ is at least $d$
from every point in $F$. 
This can be seen by contradiction: if any point in $C$ were closer
than $d$, then it would not have been picked as the furthest point
ahead of $q$. 

This provides a witness for the cost of the clustering. 
There are two cases: 
(i) some point in $C \cup \{q\}$ is assigned to some center in $F$
under the optimal clustering;
or (ii)
some pair of points in $C \cup \{q\}$ are assigned to the same center
under the optimal clustering. 
If (i) holds, then immediately we have that the cost of the optimal clustering
is at least $d$. 
If (ii) holds, then we have two points separated by distance $d$ in
the same cluster.
Then the cost of the optimal clustering is at least $d/2$, 
which happens when their cluster center is distance $d/2$ from each
(it cannot be closer to both without violating the triangle
inequality). 

On the other hand, since $q$ is the furthest from its closest center
in $C \cup F$, it follows that {\em all} points are within $d$ of their
closest center, and therefore the cost of the clustering found by the
algorithm is $d$, which is at most twice the optimal cost. 
\end{proof}

This indicates that we can find a 2-approximation for the stable
extension by considering all subsets of the current clustering, and
applying the above lemma. 
However, for subsets larger than $k/2$, we can simply work with the
(approximate) solution that switches out all $k$ centers for the
optimal $k$-center clustering on the current set: this at most doubles
the changeout cost, while improving the clustering cost. 
Thus, we only need to consider subsets of size at most $k/2$.


\section{Concluding Remarks}

We have introduced and formalized a model of stability and fit, where
the benefits of an improved solution have to be traded off against the
costs of changing from the current status quo. 
From this motivation, we have provided stable extensions of sampling
and optimization algorithms, and demonstrated their ability in
practice. 

Our work was inspired by the need to improve the performance
of an  application management system.  
Our ongoing work is to extend our comparison to these scenarios, and
evaluate against previously deployed ad-hoc approaches.

\bibliographystyle{plain}
\bibliography{cycle}



\appendix
\section{Incremental stable PPS}  \label{moredynpps}
Section~\ref{dynpps} provided an overview of our approach to incremental
 (``dynamic'') maintenance of an $\alpha$-stable PPS distribution
under modifications of the weight of
  a single entry at a time.

We now provide full details of the data structures we use.
To gain intuition, we first consider the special case ($a=0$) of
 maintaining a PPS distribution. 
Conceptually,  entries are maintained in a sorted order according to
their weight and we also track
the respective $\tau$ value.  
We use a structure, such as Fenwick tree~\cite{Fenwick:1994},  that
supports prefix sums, lookups, insertions and deletions in $O(\log n)$ time.
When an entry is modified, it is deleted from its previous position
 and reinserted with its new weight. 
If the weight decreased (increased), the threshold can only decrease
(respectively, increase).   
The respective new $\tau$ (new solution of \eqref{ppstau}) can be computed
with a logarithmic number of lookups.
An actual sample can be maintained using permanent random numbers
(PRN).  
This is facilitated by
maintaining another sorted order according to the ratios $w_i/u_i$
(where $u_i$ is the PRN of entry $i$).  The sample includes all
entries with $w_i/u_i \geq \tau$.  For this, we use a data structure
that supports insertions, deletions, and lookups (of the smallest
ratio that is larger than a value) in $O(\log n)$ time such as a self
adjusting binary tree \cite{SleatorTarjan:1985}.

In the case of  maintaining a stable sample for general $a$,
probabilities in the stable sample only need to be modified when the ratio
$w_i/p_i$ of a modified input entry is outside the current
$(\underline{\tau},\overline{\tau})$ range.  
In the batch algorithm we worked with the ordered list of ratios
$w_i/p_i$.  
The prefix and suffix of this list formed PPS samples, and thus, within
the suffix and prefix of the ordered ratios list, the ratio order
corresponds to the weights order. 
 But in the middle range, for ratios in
$(\underline{\tau},\overline{\tau})$, the order of ratios may not
correspond to the order of weights.

A critical observation for our analysis is that once two entries are
``processed together'' in a prefix or suffix of the sorted ratios
list, and thus the heavier one precedes the lighter one, their
relative position remains the same in the sorted ratios list,
even if they are no longer members of the prefix or suffix (have
ratios below $\underline{\tau}$ or above $\overline{\tau}$).   This
situation only changes, i.e. $w_i > w_j$ and $w_i/p_i < w_j/p_j$
when the weight of at least one of the entries is modified.
Intuitively, in our analysis, we charge the cost of ``correcting''
this order when an entry works its way back into a prefix or suffix to the
modifications in the input.

 The main structure that we use to support this analysis is what we
 call a {\em stretch},  which maintains a subset
 of entries.  The stretch supports the following operations in $O(\log
 n)$ time:
   deletion of an entry, splitting a stretch into two stretches that
 include all entries with weights that are respectively below and
 above some value, merging two stretches with non-overlapping weight
 ranges, performing a lookup of the smallest/largest weight above some
 value, and returning the weight of all entries that are below or above a
certain value.   
Consequently, two stretches with overlapping weight ranges can be merged in time
proportional to $\log(n)$ times the number of interleaved segments in their
joint sorted order.

 Instead of treating entries as single entities, they are grouped to stretches.
    All entries currently in the same stretch share the same ratio
$\tau = w_i/p_i$ (which we refer to as the threshold of the stretch)
or else have $w_i>\tau$ and $p_i=1$.
We thus assign a ratio to a stretch as the 
range that spans its threshold and its largest weight. 
All stretches  are maintained in a
structure which supports insertion, deletion, and accessing stretches
with maximum or minimum threshold in $O(\log n)$ time. 
This structure corresponds to the ``sorted list of ratios'' in the
batch algorithm.

When an entry is modified, it is deleted from its current stretch and
inserted as a single-entry stretch. 
To recompute the distribution in response to the modification,
we apply an extension of our batch algorithm that can manipulate
stretches:
 If the update produces a higher maximum ratio or lower minimum ratio
 and $\alpha$-stability is violated,
 we obtain new values for $\underline{\tau}$ and
$\overline{\tau}$  while processing stretches  by decreasing and
increasing ratios, as done (for single entries) by the batch algorithm.

 As the prefix (and suffix) of the ``sorted'' ratios
are processed in decreasing ratio order,  stretches are merged  into a single
prefix (or a single suffix) stretch.  We stop 
when
the difference of squares is exactly $a$.  The cost of merging
stretches depends on the number of interleaved sorted segments of
stretches, and each segment can be ``charged'' to a modification of an item.

\end{document}